\documentclass{sig-alternate}
\usepackage{multirow}

\usepackage[noend]{algorithm2e}
\usepackage{amsmath}
\usepackage{amssymb}
\usepackage{stmaryrd}
\usepackage{etoolbox}
\usepackage{multirow}     
\usepackage[usenames,dvipsnames]{color}
\usepackage{tikz}
\usetikzlibrary{arrows,calc,shapes}

\usepackage{subcaption}
\usepackage{varwidth}
\usepackage{listings}
\patchcmd{\maketitle}{\@copyrightspace}{}{}{}

\usepackage{hyperref}
\newtheorem{definition}{Definition}
\newtheorem{lemma}{Lemma}
\newtheorem{theorem}{Theorem}

\newcommand{\todo}[1]{\textcolor{red}{{#1}}}
\newcommand{\td}[1]{\todo{[#1]}}
\newcommand{\bytd}[1]{\textcolor{orange}{{#1}}}
\newcommand{\hide}[1]{}
\newcommand{\fun}[1]{\mathtt{{#1}}}

\newcommand{\FF}{\mathsf{ff}}
\newcommand{\DEC}{\mathsf{decision}}

\newcommand{\bbfN}{\mathbb{N}}
\newcommand{\pthof}[1]{\langle #1 \rangle}

\newcommand{\procs}[0]{\mathcal{P}}

\newcommand{\proccalls}[1]{\mathcal{PC}_{#1}}
\newcommand{\params}[0]{\mathit{PPN}}
\newcommand{\inlngsof}[1]{\mathit{inln(#1)}}
\newcommand{\cfg}[0]{\mathit{cfgc}}
\newcommand{\cfgof}[1]{\prog(#1)}
\newcommand{\cfgprimeof}[1]{\prog^{\#}(#1)}
\newcommand{\subst}[0]{\sigma}
\newcommand{\substof}[1]{\subst(#1)}

\newcommand{\semof}[1]{\llbracket #1 \rrbracket}
\newcommand{\pdatrans}[5]{(#1, [#2;#3/#4], #5)}
\newcommand{\langof}[1]{L(#1)}
\newcommand{\prog}[0]{\mathit{prog}}
\newcommand{\pacof}[1]{\mathit{PAC}(#1)}
\newcommand{\pacofed}[0]{\pacof{\epsilon, \delta}}
\newcommand{\probof}[2]{\mathit{Prob}_{#1}[#2]}
\newcommand{\badpaths}[0]{\mathcal{B}}

\newcommand{\bd}{d}
\newcommand{\be}{e}

\newcommand{\lift}[2]{{#1}^{\langle {#2} \rangle}}
\newcommand{\MEM}{\mathit{Mem}}
\newcommand{\EQ}{\mathit{Equ}}
\newcommand{\cpachecker}{\textsc{CPAchecker}\xspace}
\newcommand{\cbmc}{\textsc{CBMC}\xspace}
\newcommand{\libalf}{\textsc{libalf}\xspace}
\newcommand{\libamore}{\textsc{libAMoRE++}\xspace}
\newcommand{\pacman}{\textsc{Pac-Man}\xspace}

\newcommand{\crest}{\textsc{Crest}\xspace}
\newcommand{\fp}{\mathcal{X}_{\mathit{FP}}}

 \patchcmd{\maketitle}{\@copyrightspace}{}{}{}
\ifdef{\longversion}{
\newcommand{\whenlongversion}[1]{#1}
}{
\newcommand{\whenlongversion}[1]{}
}

\newtheorem{corollary}{Corollary}

\begin{document}
\title{PAC Learning-Based Verification and Model Synthesis}
\numberofauthors{7}
\author{
\alignauthor Yu-Fang~Chen\\
\affaddr{Academia Sinica}
\and
\alignauthor Chiao~Hsieh\\
\affaddr{Academia Sinica}\\
\affaddr{National Taiwan University}
\and
\alignauthor Ond\v{r}ej~Leng\'{a}l\\
\affaddr{Academia Sinica}\\
\affaddr{Brno University of Technology}
\and
\alignauthor Tsung-Ju Lii\\
\affaddr{Academia Sinica}\\
\affaddr{National Taiwan University}
\and
\alignauthor Ming-Hsien~Tsai\\
\affaddr{Academia Sinica}
\and
\alignauthor Bow-Yaw~Wang\\
\affaddr{Academia Sinica}
\and
\alignauthor Farn~Wang\\
\affaddr{National Taiwan University}
 }

\maketitle

\begin{abstract}


We introduce a~novel technique for verification and model synthesis of
sequential programs.
Our technique is based on learning a~regular model of the set of feasible paths
in a~program, and testing whether this model contains an incorrect behavior.
Exact learning algorithms require checking equivalence between the model and
the program, which is a~difficult problem, in general undecidable.
Our learning procedure is therefore based on the framework of \emph{probably
approximately correct} (PAC) learning, which uses sampling instead and provides
correctness guarantees expressed using the terms \emph{error probability} and
\emph{confidence}.
Besides the verification result, our procedure also outputs the model with the
said correctness guarantees.
Obtained preliminary experiments show encouraging results, in some cases
even outperforming mature software verifiers.

%

\end{abstract}

\vspace{-2mm}
\section{Introduction}
\label{section:introduction}


Formal verification of software aims to prove software properties
through rigorous mathematical reasoning. Consider, for example, 
the C statement \texttt{assert(x > 0)} specifying that the value of the variable
\texttt{x} must be positive. If~the assertion is formally verified,
it \emph{cannot} be violated in any possible execution during runtime. 
Formal verification techniques are, however, often computationally expensive. 
Although sophisticated heuristics have been developed to improve
scalability of the techniques, formally verifying real-world software is still
considered to be impractical.

\hide{
The complexity of current programs and the associated state space explosion
make the programs hard to be analyzed precisely, and over-approximating analyses
in abstract domains often flood users with false warnings.
Unsound analyses, on the other hand, show good scalability while
keeping the number of false warnings low~\cite{beyer:svcomp14}.
However, in case unsound analyses do not find and error, they mostly do not
provide any statistical guarantees of correctness of the program.
Statistical guarantees may be needed in the development of safety-critical
systems, for example during certification of software used in the aviation
industry~\cite{airborne-safety}.
}


A~common technique to ensure quality in industry is software
testing.
Errors in software can be detected by exploring different software behaviors
via injecting various testing vectors.
Testing cannot, however,
guarantee software is free from errors. Consider again the assertion
\texttt{assert(x > 0)}.
Unless all system behaviors are
explored by testing vectors, it is unsound to conclude that the
value of \texttt{x} is always positive. Various techniques have been
proposed to improve its coverage, but it is an inherent feature of software
testing that it cannot establish program properties conclusively.


\enlargethispage{2.4mm}

In this paper, we propose a novel learning-based approach that aims to balance
scalability and coverage of existing software engineering
techniques. In order to be scalable, as for software testing, our
new technique explores only a~subset of all program behaviors.
Moreover, we apply
machine learning to generalize observed program behaviors for better
semantic coverage.
Our technique allows software engineers to combine scalable testing with
high-coverage formal analyses and improve the quality assurance process.
We hope that this work reduces the dichotomy between
formal and practical software engineering techniques.

In our technical setting, we assume programs are annotated with
program assertions. A program assertion is a~Boolean expression
intended to be true every time it is encountered during program execution.
Given a
program with assertions, our task is to check whether all
assertions evaluate to true on all possible executions. In principle,
the problem can be solved by examining all program executions. It is,
however, prohibitive to inspect all executions exhaustively since there may
be infinitely many of them.
One way to simplify the analysis is to group the set of program executions to paths of a control flow graph.



\enlargethispage{3.5mm}

A \emph{control flow graph} (CFG) is derived from the syntactic structure of a
program source code. Each execution of a program corresponds to a path
in its control flow graph. One can therefore measure the completeness
of software testing on CFGs. Line coverage, for instance,
gives the ratio of explored edges in the CFG of the tested
program, while branch coverage is the ratio of explored
branches of this CFG. Note that such syntactic measures of code
coverage approximate program executions only very roughly. Executions that
differ in the
number of iterations in a simple program loop have the same
line and branch coverages, although their computation may be
drastically different. A~full syntactic code coverage does not
necessarily mean all executions have been explored by software testing.


Observe that program executions traversing the same path in a CFG perform the
same sequence of operations (although maybe with different values).
Consider a~path corresponding to a program execution in a CFG.
Such a path can be characterized by the sequence of
decisions that the execution took when traversing conditional
statements in the CFG. We~call a~sequence of such
decisions a \emph{decision vector}. A~decision vector is feasible if it
represents one or more (possibly infinitely many) program executions, and 
infeasible if it represents a~sequence of branching choices than can never
occur in an execution of the program.
To check whether all assertions evaluate to $\mathit{true}$ on all executions, it 
suffices to examine all feasible decision vectors and check they do not
represent any assertion-violating program execution. 
Although feasibility of a decision vector can be determined by using
an off-the-shelf Satisfiability Modulo Theories (SMT) solver, 
the set of feasible decision vectors is in general difficult to compute
exactly.
Therefore, we apply algorithmic learning, in particular the framework of
probably approximately correct learning, to construct a~regular approximation
of this~set.


\hide{
The task of program verification can be described as determining whether there
exists a~feasible decision vector such that a~path characterized by the vector
leads to an error node in the control flow graph of the program.
In general, existence of such a~vector is undecidable; moreover, the set of
feasible decision vectors is difficult to compute, in general, it may even not
be recursively enumerable.
\td{OL: now talk about inspiration in natural language processing?}
}

\hide{
We can determine feasibility of a~decision vector by constructing a~path formula
of the path corresponding to the decision vector, and then discharge the
formula using an off-the-shelf SMT solver.
}


Within the framework of \emph{probably approximately correct} (PAC) learning with queries,
learning algorithms query about target concepts to construct
hypotheses. The constructed hypotheses are then validated by
sampling. If~a~hypothesis is invalidated by witnessing a~counterexample, learning
algorithms refine the invalidated hypothesis by the witness and more
queries. If, on the other hand, a hypothesis conforms to all
samples, PAC learning algorithms return the inferred hypothesis with
statistical guarantees. In our approach, we adopt a PAC learning
algorithm with queries to infer a regular language approximation to the set of feasible decision vectors of a program.

To grasp the statistical guarantees provided by PAC learning, consider the
task of checking defects in a~large shipment using uniform
sampling. Because of the size of the shipment, it is impractical to check
every item. We instead want to know, with a~given confidence~$\delta$, if the defect probability is at 
most~$\epsilon$. This can be done by
selecting $r$ (to be determined later) randomly chosen items. If all
chosen items are good, the method reports that the defect probability
is at most~$\epsilon$. We argue the simple method can err with
probability at most $1 - \delta$. Suppose $r$ randomly chosen items
are tested without any defect, but the defect probability is, in fact,
\emph{more than}~$\epsilon$.
Under this thesis, the probability that
$r$ random items are all good is lower than $(1 - \epsilon)^r$. That is,
the method is incorrect with probability lower than
$(1 - \epsilon)^r$. Take $r$ such that
$(1 - \epsilon)^r < 1 - \delta$. The simple method
reports incorrect results with probability at most $1 -
\delta$; we equivalently say the result of the method is $\pacofed$-correct.
Using a similar argument, it can be shown that our PAC learning algorithm 
returns a~regular set approximating the set of feasible decision
vectors of the program with the error probability $\epsilon$ and confidence $\delta$
of our choice.
If the inferred set contains no decision vector representing an assertion-violating program
execution, our technique concludes the verification with statistical guarantees about correctness.
Our learning-based approach finds a balance between formal analysis
and testing. Rather than exploring program behaviors exhaustively, our
technique infers an approximation of the set of feasible decision vectors by
queries and sampling. 
Although the set of feasible decision vectors is
in general not computable, PAC learning with queries may still return a
regular set approximation of it with a quantified guarantee.
Such an approximate model
with statistical guarantees can be useful for program
verification.
With an approximate model that is $\pacofed$-correct and proved to be free from assertion violation, one can conclude that the program is also $\pacofed$-correct.
The statistical guarantees are different from syntactic code
coverages in software testing.
Recall that our application of PAC learning works over decision vectors.
Decision vectors in turn represent
program executions. When our technique does not find any assertion
violation, the statistical guarantees give software engineers a
semantic coverage about program executions. Along with conventional
syntactic coverages, such information may help software engineers
estimate the quality of software.

We implement a prototype, named \pacman (PAC learning-based Model synthesizer
and ANalyzer), of our procedure based on program verifiers \cpachecker, \cbmc,
and the concolic tester \crest.
We evaluate the prototype on the benchmarks from the recursive category of SV-COMP~2015~\cite{svcomp}.
The results are encouraging---we can find all errors that can be found by \crest.
We also provide quantified guarantee accompanied by a~faithful approximate
model for several examples that are challenging for program verifiers and
concolic testers.
This approximate model can later be reused, e.g., for verifying the same program with a~different set of program assertions.

Our contributions are summarized in the following:
\begin{itemize}
\vspace{-1.5mm}
\itemsep0mm
\item  We show the PAC learning algorithm can be applied to synthesize
  a~faithful approximate model of the set of feasible decision vectors of
  a~program.
  Such a~model can be useful in many different aspect of program verification
  (cf.~Section~\ref{section:conclusions} for details).
  We believe it is not hard to adopt our approach to handle different type of
  systems (e.g., black box systems) and to obtain approximate models on
  a~different level of abstraction (e.g., on a~function call graph).

\item We develop a verification procedure based on the approximate model obtained from PAC learning. The procedure integrates the advantages of both testing and verification. It uses  testing techniques to collect samples and catch bugs. The PAC learning algorithm generalizes the samples to obtain an approximate model that can then be analyzed by verification techniques for statistical guarantees.

\end{itemize}

\vspace{-2.5mm}
\section{Preliminaries}
\label{section:preliminaries}
Let $\mathcal{X}$ be the set of program variables and $\mathcal{F}$
the set of function and predicate symbols.
We use $\mathcal{X'}$ for the set $\{x' \mid x \in
\mathcal{X} \}$. The set $\mathcal{T}[\mathcal{X}, \mathcal{F}]$ of
\emph{transition} formulae consists of well-formed first-order logic
formulae over $\mathcal{X}, \mathcal{X'}$, and $\mathcal{F}$.
For a transition formula $f \in \mathcal{T}[\mathcal{X}, \mathcal{F}]$
and $n \in \bbfN$, we use $\lift{f}{n}$ to denote
the formula obtained from $f$ by replacing all free variables $x \in
\mathcal{X}$ and $x' \in \mathcal{X'}$ with $\lift{x}{n}$ and
$\lift{x}{n+1}$ respectively. 

We represent a~program with a~single procedure using a~control flow
graph.
(Section~\ref{section:function} extends the notion to programs with
multiple procedures and procedure calls.)
A \emph{control flow graph (CFG)} is a~graph $G=(V, E, v_i, v_r, V_e, \fp)$
where $V = V_b \cup V_s$ is a~finite set of \emph{nodes} consisting of
disjoint sets of \emph{branching} nodes~$V_b$ and \emph{sequential}
nodes~$V_s$, $v_i
\in V$ is the \emph{initial} node, $v_r \in V_s$ is the \emph{return} node, $V_e \subseteq V$ is the set of
\emph{error} nodes, $\fp\subseteq \mathcal{X}$ is the set of formal parameters, and $E$ is a~finite set of \emph{edges} such that $E \subseteq V \times
\mathcal{T}[\mathcal{X}, \mathcal{F}] \times V$ and the following conditions~hold:
\begin{itemize}
\vspace{-1.5mm}
\itemsep0mm
\item  for any branching node $v_b \in V_b$, there are exactly two nodes $v'_0,
  v'_1 \in V$ with $(v_b, f_0, v'_0), (v_b, f_1, v'_1) \in E$, where
  $f_0, f_1 \in \mathcal{T}[\mathcal{X}, \mathcal{F}]$ are transition
  formulae;
\item  for any non-return sequential node $v_s \in V_s \setminus \{v_r\}$, there is exactly one node
  $v' \in V$ with $(v_s, f, v') \in E$; and
\item  for the return node $v_r \in V_s$, there is no $v' \in V$ such
  that $(v_r, f, v') \in E$ for any $f \in \mathcal{T}[\mathcal{X},
  \mathcal{F}]$.
\vspace{-1.5mm}
\end{itemize}
\hide{
A \emph{control flow graph (CFG)} of a~program is a~tuple $G=(V_b \cup V_s, E, v_i, V_e, \fp)$
where $V_b \cup V_s$ is the (finite) set of \emph{nodes} consisting of
disjoint sets of \emph{branching} nodes~$V_b$ and \emph{sequential} nodes~$V_s$, the node $v_i
\in V_s$ is the \emph{initial} node, the set $V_e \subseteq V_b \cup V_s$ is the set of
\emph{error} nodes, $\fp\subseteq \mathcal{X}$ is the set of formal parameters, and the set $E \subseteq (V_b \cup V_s) \times
\mathcal{T}[\mathcal{X}, \mathcal{F}] \times (V_b \cup V_s)$ is the set of
\emph{edges}, where
\begin{itemize}
\item for any branching node $v_b \in V_b$, there are exactly two nodes $v'_0,
  v'_1 \in V_b \cup V_s$ with $(v_b, f_0, v'_0), (v_b, f_1, v'_1) \in E$, where
  $f_0, f_1 \in \mathcal{T}[\mathcal{X}, \mathcal{F}]$ are transition
  formulae such that $f_0 \lor f_1$ is a valid formula; and
\item for any sequential node $v_s \in V_s$, there is at most one node
  $v' \in V_b \cup V_s$ with $(v_s, f, v') \in E$ for $f \in
  \mathcal{T}[\mathcal{X}, \mathcal{F}]$.
\end{itemize}
}
We say $v'$ is a \emph{successor} of $v$ if $(v, f, v') \in E$.
Assume, moreover, that the two successors $v'_0$ and $v'_1$ of the
branching node $v$ are ordered.
Intuitively, the `1' corresponds to the \emph{if} branch and the `0'
corresponds to the \emph{else} branch.
We call $v'_0$ and $v'_1$ the
\emph{$0$-successor} and \emph{$1$-successor} of $v$ respectively.
Similarly, $f_0$ and $f_1$ are called the \emph{$0$-transition} and
\emph{$1$-transition formulae} of $v$.
Note that the definition of a~CFG allows us to describe nondeterministic
choice, which is commonly used to model the environment.
To be more specific, a~nondeterministic choice from a branching node~$v$ can be
represented by defining both the $0$-transition and $1$-transition formulae of
$v$ as $\bigwedge_{x \in
\mathcal{X}} x=x'$.


A \emph{path} in the CFG $G$ is a sequence $\pi = \pthof{v_0, f_1, v_1, f_2, v_2,
\ldots, f_m, v_m}$ such that $v_0 = v_i$ and $(v_j, f_{j+1}, v_{j+1}) \in E$ for
every $0 \leq j < m$.
The path $\pi$ is \emph{feasible} if the path formula
$\bigwedge_{k=1}^m \lift{f_k}{k}$ is satisfiable.
It is an \emph{error}
path if $v_j \in V_e$ for some $0 \leq j \leq m$. 
The task of our analysis is to check whether $G$ contains
a~feasible error~path.

A sequence $w = a_1 a_2
\cdots a_n$ with $a_j \in \{0,1\}$ for $1 \leq j \leq n$ is called a
\emph{word} over $\{0,1\}$. The \emph{length} of $w$ is $|w| = n$. The
word of length 0 is the \emph{empty} word $\lambda$. We also use
$w[j]$ to denote the $j$-th symbol $a_j$.
If $u, w$ are words over $\{0,1\}$, $u \cdot w$~denotes the
\emph{concatenation} of $u$ and~$w$.
A~\emph{language} $L$ over $\{0,1\}$ is a set of words over $\{0,1\}$.
\hide{$L^{<k} = \{ w \in L \mid | w | < k\}$. } 
\hide{
  A~\emph{language} over $\{0,1\}$ is a set of words over $\{0,1\}$.
  Let $L$ be a language over
  $\{0,1\}$. We define $L^{\leq k} = \{ w \in L \mid | w | \leq k\}$ and,
  similarly, $L^{=k} = \{ w \in L : | w | = k\}$.
  \hide{$L^{<k} = \{ w \in L \mid | w | < k\}$. } 
}

\hide{
For a word $d = d_1 d_2 \dots d_n$, we use $d[j]$ to denote the symbol $d_j$ if
$0 \leq j \leq n$; otherwise $d[j]$ is undefined.
}


We introduce the function $\DEC$ that maps a~path~$\pi$ of~$G$ to a~sequence of
decisions made in the branching nodes traversed by $\pi$.
Formally, $\DEC$ is a~function from paths to words over $\{0,1\}$ defined recursively as follows:
$\DEC(\pthof{v_0, f_1, v_1, f_2, v_2, \ldots, f_m, v_m}) = \DEC(\pthof{v_0, f_1}) \cdot\DEC(\pthof{v_1, f_2, v_2, \ldots, f_m, v_m})$ such that
\begin{eqnarray*}
  \DEC(\pthof v)     \!\!\!&=& \lambda \\
  \DEC(\pthof{v, f}) \!\!\!&=&
  \!\!\!\left\{
    \begin{array}{ll}
      \lambda & \textmd{if } v \in V_s,\\
      0 & \textmd{if } v \in V_b \textmd{ and } \\
      & f \textmd{ is the
        $0$-transition formula of } v,\\
      1 & \textmd{if } v \in V_b \textmd{ and }\\ 
      & f \textmd{ is the
        $1$-transition formula of } v.\\
    \end{array}
  \right.
\end{eqnarray*}
For a path $\pi$, $\DEC (\pi)$ is the \emph{decision vector} of $\pi$.
We lift $\DEC$ to a set of paths $\Pi$ and define \emph{decision vectors} of
$\Pi$ as $\DEC (\Pi) = \{ \DEC (\pi) \mid \pi \in \Pi \}$.

A \emph{finite automaton} (with $\lambda$-moves) $A$ is a
tuple $A = (\Sigma, Q, q_i, \Delta,
F)$ consisting of a finite \emph{alphabet} $\Sigma$, a~finite set of
\emph{states} $Q$, an 
\emph{initial} state $q_i \in Q$, a~\emph{transition relation} $\Delta
\subseteq Q \times (\Sigma \cup \{ \lambda \}) \times Q$, and a~set of
\emph{accepting} states $F \subseteq Q$.
A~transition $(q, \lambda,
q') \in \Delta$ is called a \emph{$\lambda$-transition}.
A word $w$ over $\Sigma$ is
%
%
%
\emph{accepted} by $A$ if there are
states $q_0, \dots, q_m \in Q$ and symbols (or
$\lambda$'s) $a_1, \dots, a_m \in (\Sigma \cup \{ \lambda \})$,
such that $w = a_1 \cdots a_m$,
for every $0 \leq j < m$ there is a~transition $(q_j, a_{j+1},  q_{j+1}) \in \Delta$,
and further $q_0 = q_i$ and $q_{m} \in F$.
The \emph{language} of~$A$ is defined as $L (A) =
\{ w \mid w \textmd{ is accepted by } A \}$. A~language~$R$ is \emph{regular} if
$R = L (A)$ for some finite automaton~$A$. The finite automaton~$A$ is
\emph{deterministic} if its transition relation is a function from $Q
\times \Sigma$ to $Q$. For any finite automaton~$A$, there exists
a~deterministic finite automaton (DFA) $B$ such that $L (A) = L (B)$.


\hide{
A finite automaton with $\epsilon$-transitions (FA-$\epsilon$) is
a~tuple $A_{\epsilon} = (\Sigma, Q, q_i, \Delta_{\epsilon}, F)$ where $\Sigma$,
$Q$, $q_i$, and $F$ have the same meaning as in the definition of an FA (with
the additional constraint that $\epsilon \not\in \Sigma$), but the transition
relation
$\Delta_{\epsilon} \subseteq Q \times (\Sigma \cup \{\epsilon\}) \times Q$ can
also contain the so-called $\epsilon$-transitions.
Note that any FA-$\epsilon$ $A_{\epsilon}$ can be transformed into an FA $A$
using the standard textbook algorithm for removing $\epsilon$-transitions.
We define the language of $A_{\epsilon}$ to coincide with the language of $A$.
}

A~\emph{pushdown automaton} (PDA) is a~tuple $P = (\Sigma, Q, \Gamma, q_i,
\Delta, F)$ where $\Sigma$ is a~finite \emph{input alphabet}, $Q$ is a~finite
set of \emph{states}, $\Gamma$ is a~finite \emph{stack alphabet}, $q_i \in Q$
is the \emph{initial state}, $F \subseteq Q$ is the set of \emph{final states},
and $\Delta \subseteq Q \times (\Sigma \cup \{\lambda\}) \times (\Gamma \cup
\{\lambda\})\times (\Gamma \cup \{\lambda\}) \times Q$ is a~\emph{transition
relation}.
We use $\pdatrans q a b c {q'}$ to denote the transition $(q, a, b, c, q')$,
and we sometimes simplify $\pdatrans q a \lambda \lambda {q'}$ to
$(q, a, q')$.
We define a~\emph{configuration} of $P$ as a~pair $(q, \gamma) \in Q \times
\Gamma^*$.
A~word~$w$ over $\Sigma$ is \emph{accepted} by $P$ if there exists a~sequence of
configurations $(q_0, \gamma_0), \ldots, (q_m, \gamma_m) \in Q \times \Gamma^*$
and a~sequence of symbols (or $\lambda$'s) $a_1, \ldots, a_m \in (\Sigma \cup
\{\lambda\})$, such that $w = a_1 \cdots a_m$, $q_0 = q_i$, $\gamma_0 =
\epsilon$, $q_m \in F$, and for every $0 \leq j < m$ it holds that
there are some $b_j, b_{j+1} \in (\Gamma \cup \{\lambda\})$ and $\gamma'_j,
\gamma'_{j+1} \in \Gamma^*$ such that $\gamma_j = b_j \gamma'_j, \gamma_{j+1} =
b_{j+1} \gamma'_{j+1}$, and there is $\pdatrans {q_j} {a_{j+1}} {b_j}
{b_{j+1}} {q_{j+1}} \in \Delta$.
The language of $P$ is defined as $\langof P = \{w \mid w \text{~is accepted
by~} P\}$.

\vspace{-2mm}
\section{Overview}
\label{section:overview}
In this section, we give an overview of our verification procedure.
Let $G$ be a CFG of a program. Our goal is to check whether there is a
feasible error path in $G$.
More concretely, consider the set $\Pi$ of feasible paths
in $G$ and the set $\badpaths$ of error paths in
$G$.
We~call the languages $\DEC (\Pi)$ and $\DEC (\badpaths)$ over the alphabet
$\{ 0, 1 \}$ as \emph{feasible decision
vectors} and \emph{error decision vectors} respectively.
The program is correct if the intersection $\DEC
(\Pi) \cap \DEC (\badpaths)$ is empty, i.e., if $G$ contains no feasible error path.

\hide{
We therefore would like to check if $\DEC (\badpaths) \cap \DEC (\Pi)$
is empty.
}

\begin{figure}[t]
\scalebox{0.85}{


\begin{tikzpicture}[node distance = 2cm, auto]
\tikzstyle{decision} = [diamond, draw, fill=blue!20, aspect=3,text width=2.5cm, text badly centered, inner sep=0pt]
\tikzstyle{block} = [rectangle, draw, fill=blue!20, 
     text centered, rounded corners, minimum height=0.8cm]

\tikzstyle{line} = [draw, -latex']

    \draw[rounded corners,fill=green!20] (-2cm,1cm) rectangle +(8cm,-1.3cm) ;
	\node[text centered,text width=6cm] at (2cm,0.3cm) {\textbf{PAC Automata Learning Algorithm}\\(Section~\ref{section:learning})};

    \draw[rounded corners,fill=yellow!20] (-2cm,4.5cm) rectangle +(8cm,-2.5cm) ;

	\node at (2,4) {\bf Mechanical Teacher};
    
    \node [block,text width=3cm] at (0,3) {Resolving Membership Queries (Section~\ref{section:membership})};

    \node [block,text width=3cm] at (4,3) {Resolving Equivalence Queries by Sampling (Section~\ref{section:sampling})};

	\node[text width=7cm] at (1.5,4.7){Found a feasible error decision vector};
	\node[] at (4.5,5.1){The system is $\pacofed$-correct};
    
    \path [line] (-0.5,1) -- node {$\MEM(w)$} (-0.5,2.4);
    \path [line] (0.5,2.4) -- node {yes/no} (0.5,1);

    \path [line] (3.5,1) -- node {$\EQ(C)$} (3.5,2.4);
    \path [line] (4.5,2.4) -- node {counterexample} (4.5,1);





    \path [line] (3.8,3.65) -- (3.8,4.7) -- (3.4, 4.7);
    \path [line] (4.1,3.65) -- (4.1,5);
    \path [line] (-2,0) -- (-2.3,0) -- (-2.3, 4.7)--(-2.1,4.7);

\end{tikzpicture}
}
\vspace{-2mm}
\caption{Components of our verification procedure}
\vspace{-6mm}
\label{figure:flowchart-hl}
\end{figure}
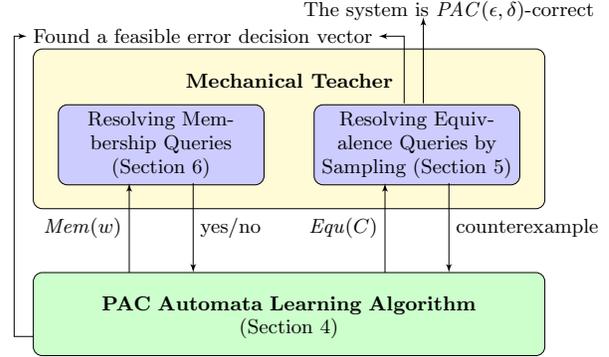

Representation of the language $\DEC (\Pi)$ of all feasible
decision vectors in $G$ is not so easy.
In general, this language may not be regular or even computable.
In our procedure, we construct a~candidate finite automaton~$C$ that
approximates $\DEC(\Pi)$, 
the set of feasible decision vectors of $G$.
We infer $C$ using a~\emph{probably approximately correct} (PAC) online
automata learning algorithm~\cite{angluin:learning1987}.
The use of PAC learning provides us with statistical guarantees about the
correctness of $C$---we can claim that $C$ is $\pacofed$-correct, i.e., with
\emph{confidence}~$\delta$, the deviation of $\langof C$ from $\DEC(\Pi)$
is less than~$\epsilon$ (we~give a~proper explanation of the terms in
Section~\ref{section:learning}).

On the other hand, it is straightforward to convert $G$ to a finite automaton $B$
accepting the set of all error decision vectors~$\DEC (\badpaths)$.
Intuitively, states of~$B$ correspond to nodes of $G$, the initial state of~$B$
corresponds to the initial node of $G$, and accepting states of~$B$ correspond
to $G$'s error nodes.
An~edge from a~sequential node is translated to a
$\lambda$-transition. For a branching node, the edges to its $0$- and
$1$-successors are translated to transitions over symbols $0$ and
$1$ respectively (cf.~Section~\ref{section:error-branch-decisions}). 

A~high-level overview or our learning procedure is given in
Figure~\ref{figure:flowchart-hl} (the procedure is similar in structure to the
one of~\cite{angluin:learning1987}).
It~consists of two main components:
The learning algorithm asks the teacher two kinds of questions:
\emph{membership} (``Is a~given decision vector feasible?'') and
\emph{equivalence} (``Is a~given candidate finite automaton
$\pacofed$-correct?'') queries.
The teacher resolves the queries, at the same time observing whether some of
the tested decision vectors corresponds to a~feasible error path.
By posing these queries, either the learning algorithm iteratively constructs
a~$\pacofed$-correct approximation of $\DEC(\Pi)$ or our procedure finds
a~feasible error decision vector.


\hide{However, in our setting, we do not have a unique target; any language
satisfying the two conditions is qualified. This is very similar to the setting
of \emph{learning-based compositional
verification}~\cite{cobleigh:learning2003,chen:learning2009}, where the learning algorithm is
modified to return a~counterexample in case the system is buggy.
We modified the learning algorithm in a similar way. At the bottom of Figure~\ref{figure:flowchart-hl}, we add two tests to handle the case when the system is buggy, i.e., $L(C)\cap \DEC(\Pi) \neq \emptyset$.
\td{OL: I don't understand the last sentence, YFC: try again, OL: I still don't understand the last sentence, YFC:yet another try.}
}

As with other online learning-based techniques, we need to devise
a~mechanical teacher that answers queries from the learning
algorithm.
Checking membership queries (i.e., membership in the set $\DEC(\Pi)$ of
feasible decision vectors) is relatively
easy---for example, given a~decision vector $d$, we obtain its corresponding
path $\pi$ by unfolding the CFG~$G$ according to $d$, and use
an off-the-shelf solver
to decide whether $\pi$ is feasible or not
(cf.~Section~\ref{section:membership}).


When the automata learning algorithm infers a candidate finite automaton $C$,
we need to check whether $\langof C$ approximates $\DEC (\Pi)$, i.e., whether $C$ is
$\pacofed$-correct.
Since we cannot compare $\DEC(\Pi)$ with
$\langof C$ directly, we employ the sampling-based approximate equivalence technique of
PAC learning.
While generally unsound, the technique still provides statistical guarantee
about the correctness of the inferred model (details are given in
Section~\ref{section:sampling}).

\vspace{-2mm}
\section{PAC Automata Learning}
\label{section:learning}
Here we explain the PAC automata learning algorithm that we use to find an approximation to $\DEC(\Pi)$. Classical PAC automata learning algorithm cannot be used directly for the purpose of program verification. It has to be modified to handle the case when the program contains an error.
The classical PAC automata learning algorithm was obtained from modifying the requirement of the exact automata learning algorithm~\cite{angluin:queries1987}. Our modification follows the same route.
In this section, we first describe the classical ``exact'' automata learning algorithm of regular languages and then describe how to modify it for verification.
Then we explain how to relax the requirement of an exact automata learning algorithm to infer an approximation to $\DEC (\Pi)$.

\subsection{Exact Learning of Regular Languages}
Suppose $R$ is a \emph{target} regular language such that its description is
not directly accessible.
\emph{Automaton learning} algorithms~\cite{angluin:learning1987,rivest:inference1993,kearns:introduction1994,bollig:angluin2009} infer automatically a~finite automaton~$A_R$ recognizing $R$.
The setting of an online learning algorithm
assumes a~\emph{teacher} who has access to $R$ and can answer the
following two types of queries:
\begin{itemize}
\item Membership query $\MEM(w)$: is the word $w$ a~member of $R$, i.e., $w \in R$?
\item Equivalence query $\EQ(C)$: is the language of the finite automaton~$C$ equal to
$R$, i.e., $L (C) = R$?
If not, what is~a counterexample to this equality
  (a~word in the symmetric difference of $L(C)$ and $R$)?
\end{itemize}
The learning algorithm will then construct a finite automaton $A_R$ such
that $L (A_R) = R$ by interacting with the teacher.
Such an algorithm works iteratively: In each iteration, it
performs membership queries to get information about $R$ from the teacher.
Using the results of the queries, it proceeds by
constructing a~candidate automaton $C$ and, finally, makes an
equivalence query $\EQ (C)$. If $L (C) = R$, the algorithm terminates
with $C$ as the resulting finite automaton~$A_R$.
Otherwise, the teacher returns a word $w$ distinguishing $L (C)$ from the
target language $R$. The learning algorithm uses $w$ to
modify the conjecture for the next iteration.
The mentioned learning algorithms are guaranteed to find a~finite automaton~$A_R$
recognizing $R$ using a~number of queries polynomial to the
number of states of the minimal DFA recognizing $R$. In the rest of the paper, we denote ``online automata learning'' simply as ``automata learning''.

\subsection{Learning for Program Verification}
\label{section:modify_learning}
Under the context of program verification, it may be the case that $\DEC(\Pi)\cap L(B)\neq \emptyset$; in such a case, our procedure should return a feasible error path in the program. This is very similar to the setting
of \emph{learning-based
verification}~\cite{cobleigh:learning2003,chen:learning2009}, where the learning algorithm is
modified to return a~counterexample in case the system contains an error.
We modified the used learning algorithm in a similar way. 
To be more specific, when the classical learning algorithm poses an equivalence query $\EQ(C)$, we first check whether there exists a decision vector $c$ such that $c\in L(C)\cap L (B)$ and then test if $c\in \DEC(\Pi)$.
\begin{enumerate}
\item In case that the two tests identified a decision vector $c$ such that $c
 \in L(C)\cap L(B)$ and $c \not \in \DEC(\Pi)$, then $c$ is in the difference
 of $L(C)$ and $\DEC(\Pi)$ and hence a valid counterexample for the classical
 learning algorithm to refine the next conjecture automaton $C$. 
\item In case that the two tests identified a decision vector $c$ such that $c
 \in L(C)\cap L(B)$ and $c \in \DEC(\Pi)$, then $c$ is a feasible error
 decision vector and we report $c$ to the user.
\item In case that $L(C) \cap L(B) =\emptyset$, the modified learning
 algorithm poses an equivalence query $\EQ(C)$ to the teacher.
\end{enumerate}

Given a teacher answers membership and equivalence queries about $\DEC(\Pi)$, the modified automata learning algorithm has the following properties.

\begin{lemma}
Assume $\DEC(\Pi)$ is a regular set. The modified automata learning algorithm eventually finds a counterexample $c \in L(B)\cap \DEC(\Pi)$ when $L(B)\cap \DEC(\Pi)\neq \emptyset$. It eventually finds a finite automaton recognizing $\DEC(\Pi)$ when $L(B)\cap \DEC(\Pi)= \emptyset$.
\end{lemma}

Observe that when the program does not contain any error, the behavior of the modified learning algorithm is identical to the classical one and hence is still an exact automata learning algorithm. Next we explain how to relax the requirements of the exact automata learning algorithm to obtain a PAC automata learning algorithm that is suitable for program verification.

\vspace{-0.0mm}
\subsection{\hspace{-0.2cm}Probably Approximately Correct Learning}\label{sec:pac-learning}
\vspace{-0.0mm}

The techniques for learning automata we just discussed in the previous section assume a teacher who has the ability to answer equivalence queries.
Such an assumption is, however, invalid in our procedure.
Checking $\DEC(\Pi) = \langof C$ can be undecidable.
\hide{Consider, for instance, a~finite state machine in a~black-box.
The teacher is in this case unable 
to answer equivalence queries since the structure of the
machine is hidden from him. Even in the case the teacher has access to the description
of the target language,
the complexity of answering equivalence queries can be prohibitive
(e.g., in the case of nondeterministic finite automata, it is
a~\textbf{PSPACE}-complete problem).
\hide{
Assuming such a teacher can be unrealistic sometimes.
}

Compared to equivalence queries, membership queries can be answered
rather easily even for black-box representations of the target language.
The teacher simply needs to send the inputs to the black-box and observe the output.
It is therefore natural to attempt to substitute equivalence testing by
observing results of membership queries of a set of sampled words in the two
regular languages.
This approach is clearly unsound---it is in general impossible to check
equivalence of a~pair of infinite sets by sampling a~finite number of elements.
On the other hand,}
Angluin showed in~\cite{angluin:queries1987}
that if we substitute equivalence queries with sampling, we can still make
statistical claims about the difference of the inferred set and the target set.

Assume that we are given a~probability distribution~$D$ over the elements of
a~universe $\mathcal{U}$, and a~hypothesis in the form $$\probof{w \in
\mathcal{U}|_D}{\neg\varphi(w)} \leq \epsilon.$$
In the hypothesis, the term $\probof{w \in \mathcal{U}|_D}{\neg\varphi(w)}$ denotes the
probability that the formula~$\varphi(w)$ is invalid for $w$ chosen randomly from
$\mathcal{U}$ according to the distribution $D$.
We call $\epsilon$ the \emph{error} parameter
and use the term \emph{confidence} to denote the
least probability that the hypothesis is correct.
We say that $\varphi(w)$ is $\pacofed$\emph{-valid} if $\probof{w \in
\mathcal{U}|_D}{\neg\varphi(w)} \leq \epsilon$ with confidence $\delta$.

In the setting of automata learning, the considered universe is $\Sigma^*$ and the target regular
language is $R \subseteq \Sigma^*$.
The task of an equivalence query $\EQ(C)$ is changed from checking
\emph{exact} equivalence, which we can express as checking that $\forall w \in \Sigma^*:
w \notin R \ominus \langof C$ (we use $\ominus$ to denote the symmetric difference operator), to
checking \emph{approximate equivalence}, i.e., checking whether the formula
$\varphi(w) = w \notin R \ominus \langof C$ is $\pacofed$-valid.
In other words, we check whether $\probof{w \in \Sigma^*|_D}{w \in R \ominus \langof C}
\leq \epsilon$ with confidence $\delta$.
For a~fixed $R$ and a~candidate~$C$, we say that $C$ is
$\pacofed$\emph{-correct} if $w \notin R \ominus \langof C$ is $\pacofed$-valid.

The teacher checks the $\pacofed$-correctness of $C$ by picking $r$ samples according
to $D$ and testing if all of them are not in $R \ominus \langof C$.
For the $i$-th equivalence query of the learning algorithm, the number of samples $q_i$ needed to establish that $C$ is
$\pacofed$-correct is given by Angluin in~\cite{angluin:learning1987} as
%
\begin{equation}
\label{eq:num-samples}
q_i = \left\lceil \frac{1}{\epsilon} \left(\ln \frac{1}{1-\delta}+i \ln 2 \right)\right\rceil .
\end{equation}
Since the inferred set $C$ is guaranteed to be $\pacofed$-correct, this approach is termed \emph{probably approximately correct}
(PAC) learning~\cite{valiant:theory1984}.

\vspace{-2mm}
\section{Resolving Equivalence Queries by Sampling}
\label{section:sampling}
The current section discusses how to design a mechanism that the
teacher can use for equivalence queries to provide the $\pacofed$-correctness
guarantee, as defined in Section~\ref{section:learning}.
Given a~probability distribution $D$ over the set of feasible decision vectors
$\DEC(\Pi)$, we can use~$D$ to give a~formal definition of the quality of
a~candidate automaton $C$.
In particular, we use as a~measure the probability with which a~decision vector
chosen randomly from $\DEC(\Pi)$ (according to the distribution $D$) is
contained in $C$.


A sampling mechanism offering such a distribution must satisfy the following conditions:
\begin{enumerate}
\item  Only decision vectors in $\DEC(\Pi)$ are sampled.
\item  The samples are \emph{independent and identically distributed} (IID),
  i.e., the distribution is fixed and the probability of sampling a particular
  element does not depend on the previously picked samples.
\end{enumerate}
In this paragraph, we introduce the \emph{random input sampling} mechanism.
We treat all nondeterministic choices and formal parameters of the program as input variables
and assume that all input variables are over finite domains.
Each set of initial values of input variables yields a path in the CFG
of the program.
Based on this observation, random input sampling works by
(1)~picking uniformly at random a~set of initial values for input variables of
the program and then
(2)~obtaining the corresponding decision vector by traversing the CFG of the
program using the picked values.
The sampling mechanism forms a~distribution such that the probability of
a~decision vector~$d$ being chosen is proportional to the number of program paths
corresponding to~$d$.

The issue of random input sampling is that it suffers from the well-known fact
that coverage of input values is not a good approximation of program path
coverage.
Depending on the sizes of input domains of program variables, some paths
might have only a~negligible probability of being selected---for instance,
given two 64-bit integers \texttt{x}~and~\texttt{y}, the probability of taking
the \emph{true} branch in the test \texttt{x~==~0~\&\&~y~==~0} is equal to~$2^{-128}$.
The situation gets even worse for input variables over unbounded domains.
Even with an extremely high coverage rate of input variables' values,
many paths may still not be explored, while other are explored
repeatedly.
In order to get a~sampling mechanism with a~better distribution over program
paths, we developed a~technique that randomly explores program's paths
using a~concolic
tester~\cite{godefroid:dart2005,sen:cute2005,burnim:heuristics2008}, which is
an efficient means for exploring decision vectors corresponding to rare paths.

We describe the technique and prove its properties in the rest of this section.

\subsection{Concolic Testing}

Concolic testing is a~testing approach that explores paths in the CFG of
a~program while searching for bugs.
The algorithm begins with a decision vector generated by randomly picked input values.
Then, it finds the next decision vector by flipping some decision made in the
chosen path and obtains new input values that lead the program execution
according to the new path.
This mechanism gives rare paths a much greater chance to be explored.
The selection of which decision should be flipped depends on the used \emph{search strategy} of the tester.

In our procedure, we use the concept of a~batched sample.
A~\emph{batched sample} is defined as a~set of decision vectors of the size~$k$
(where $k$ is a~given parameter) obtained from a~concolic tester by exploring
$k$~paths using its search strategy.
We denote $D_k$ the distribution over elements of~$(\Sigma^*)^k$ obtained in
this way.
Our procedure restarts the concolic tester after taking every batched sample.
The previous point gives us the guarantee that the probability of taking each
batched sample remains the same during the execution our procedure (we assume
that the concolic tester does not keep state information between its restarts),
and that the distribution is IID and, therefore, meets condition~2 defined above.
The principal functioning of concolic testers guarantees that condition~1 is
also met.

\vspace{-0.8mm}
\subsection{Generalized Stochastic Equivalence}\label{sec:gen-stoch-eq}
\vspace{-0.0mm}

In this section, we show that our sampling mechanism using batched samples has
the property required for the $\pacofed$-correctness guarantee of the
learning algorithm given in~Section~\ref{sec:pac-learning}.

Recall that for the set of feasible decision vectors of a~program $\DEC(\Pi)$
and a~candidate automaton $C$ inferred by the learning algorithm using some
distribution $D$ over $\Sigma^*$, if the teacher gives the answer \emph{yes}
for the equivalence query $\EQ(C)$, it guarantees with confidence~$\delta$ that
\begin{equation}
\vspace{-0.2mm}
  \probof{w\in \Sigma^*|_{D}}{w \in \DEC(\Pi) \ominus \langof{C}} \leq \epsilon .
\end{equation}
Since our sampling technique uses batched samples from the universe
$\mathcal{U}_k = (\Sigma^*)^k$ w.r.t.\ the distribution~$D_k$ instead of
elements of $\Sigma^*$ and distribution~$D$, we need to change the provided
guarantee in our modification of the learning algorithm.
If our algorithm answers \emph{yes}, it guarantees that
\begin{equation}
  \label{eq:mod-guarantee}
\vspace{-0.2mm}
  \probof{S\in \mathcal{U}_k|_{D_k}}{\exists w\in S: w \in \DEC(\Pi) \ominus \langof{C}} \leq \epsilon
\end{equation}
with confidence~$\delta$ (we hereafter use the term $\pacofed$-correct to
denote this form of guarantee).

When a teacher receives an equivalence query $\EQ(C)$, it uses a~concolic tester to obtain $q_i$ (given in~(\ref{eq:num-samples})) batch samples.
For each batch sample $S$, the teacher checks if there exists a decision vector
$w\in S$ such that $w \notin \langof{C}$ (by definition $w\in \DEC(\Pi)$).
The teacher answers \emph{yes} if there is no such $w$.
Otherwise, the teacher checks if $w$ is an error decision vector and either
reports $w$ as a~feasible error decision vector or returns $w$ to the learning
algorithm to refine the next conjecture.

The following lemma shows that if we use the number $q_i$ batched samples for testing the equivalence, we
obtain the modified $\pacofed$-correctness guarantee from~(\ref{eq:mod-guarantee}).

\begin{lemma}
\label{lemma:mlstar}
Let $\epsilon$ and $\delta$ be the error and confidence parameters, and $R$ be
the target language. If no decision vector $w \notin \langof C$ is found in the $q_i$ batched samples, then it holds that $C$ is $\pacofed$-correct. 

\end{lemma}
Based on the fact that $\langof C \cap \langof B =\emptyset$ (the property of the modified learning algorithm in Section~\ref{section:modify_learning}) and the lemma above, we obtain the following corollary.
\begin{corollary}
Let $\epsilon$ and $\delta$ be the error and confidence parameters, and $R$ be
the target language. If no decision vector $w \notin \langof C$ is found in $q_i$ batched samples, then it holds that the program is $\pacofed$-correct. 
\end{corollary}

\vspace{-2mm}
\section{Resolving Membership Queries}
\label{section:membership}

In this section, we describe how a membership query $\MEM(\bd)$
in the algorithm in Figure~\ref{figure:flowchart-hl}
is discharged by the teacher.
Let $\Pi$ be the set of feasible paths of a~CFG~$G$. When the learning
algorithm asks a~membership query
$\MEM (\bd)$, the teacher needs to check whether the decision vector~$\bd$ is
in the set of feasible decision vectors $\DEC (\Pi)$.
To answer the query, the teacher first constructs a~path $\pi=\pthof{v_0, f_1, v_1,
f_2, v_2, \ldots,v_{m-1}, f_m, v_m}$ in $G$ such that 
\begin{itemize}
\vspace{-1.3mm}
\itemsep0mm
\item  there are exactly $|\bd|$ occurrences of branching nodes in
       the prefix $\pthof{v_0, f_1, v_1, f_2, v_2, \ldots,v_{m-1}}$ of $\pi$,

\item if $v_k$ is the $j$-th branching node in $\pi$, it holds that
  $\DEC(v_k,f_{k+1}) = \bd[j]$, and

\item $v_{m-1}$ is a branching node.
\vspace{-1.3mm}
\end{itemize}
Recall that $\pi$ is a~feasible path if and only if $\varphi = \bigwedge_{j=1}^m
\lift{f_j}{j}$ is satisfiable.
Therefore, the teacher can simply construct the formula $\varphi$ from the path
$\pi$ and check its satisfiability using an off-the-shelf constraint solver.

Alternatively, the teacher can check feasibility by translating the path into a
sequence of program statements (with conditions removed and substituted by
assertions on the values of the conditions) and asking a~symbolic executor or software model checker whether the final line of the constructed program is reachable.
The alternative option is easier to implement but usually suffers from some performance penalty.


\hide{
Note that other paths can also be mapped to the same branching decision as $\pi$.
One may wonder that why it is sufficient to check just the path $\pi$. The main reason is that all such paths are with the prefix $\pi$ and also with $v_{m-1}$ as the last branching node. Since all edges from sequential nodes are non-blocking, the edge $(v_{m-1},f_m, v_m)$ is the last possibility to make those paths infeasible. Therefore if $\pi$ is feasible, then all such paths are also feasible. 
}

\vspace{-2mm}
\section{Error Decision Vectors}
\label{section:error-branch-decisions}
Let $\badpaths$ be the set of \hide{(not necessarily feasible)} error paths in
a~CFG.
In this section, we show how we construct a finite automaton accepting the
set of all error decision vectors $\DEC (\badpaths)$ of the given CFG.
This automaton will later be intersected with the automaton representing the
set of feasible paths to determine whether the CFG contains a~feasible error
path.

\begin{definition}
Let $G = (V_b \cup V_s, E, v_i, v_r, V_e, \fp)$ be a~CFG.
We define the \emph{error trace automaton} for $G$ as the finite automaton $B =
(\{0, 1\}, V_b \cup V_s, v_i, \Delta_E, V_e)$ where $\Delta_E$ is defined as follows:
\begin{itemize}
\vspace{-1mm}
\itemsep0mm
\item $(v, 0, v'_0) \in \Delta_E$ if $v \in V_b\setminus V_e$, $(v, f_0, v'_0) \in E$,
  and $v'_0$ is the $0$-successor of $v$;
\item $(v, 1, v'_1) \in \Delta_E$ if $v \in V_b\setminus V_e$, $(v, f_1, v'_1) \in E$,
  and $v'_1$ is the $1$-successor of $v$;
\item $(v, \lambda, v') \in \Delta_E$ if $v \in V_s\setminus V_e$ and $(v, f, v') \in
  E$; and
\item $(v, 0, v), (v, 1, v) \in \Delta_E$ if $v \in V_e$.
\end{itemize}
\end{definition}
Informally, $B$ contains a~state for every node and a~transition for every
edge of $G$.
It reads a symbol in each state corresponding to a~branching node and performs
$\lambda$-transitions for states corresponding to sequential
nodes. For every error node, $B$ reads all remaining symbols
and accepts the input word. It is straightforward to see that
$B$ accepts exactly the set of decision vectors corresponding to
error paths in $G$.
\begin{lemma}
  Let $G = (V, E, v_i, v_r, V_e, \fp)$ be a CFG and $\badpaths$ the set
  of error paths in $G$. Let $B$ be the error trace automaton for $G$. It holds that $L (B) = \DEC (\badpaths)$.  
\end{lemma}


In Section~\ref{section:function}, we describe an extension of our procedure to
programs with procedure calls.
Because representing the set of error decision vectors using a~finite automaton
is in this setting imprecise, the section also discusses an extension that
represents the set of error paths in a~program with procedure calls using
pushdown automata.

\vspace{-2mm}
\section{The Main Procedure}
\label{section:verification}
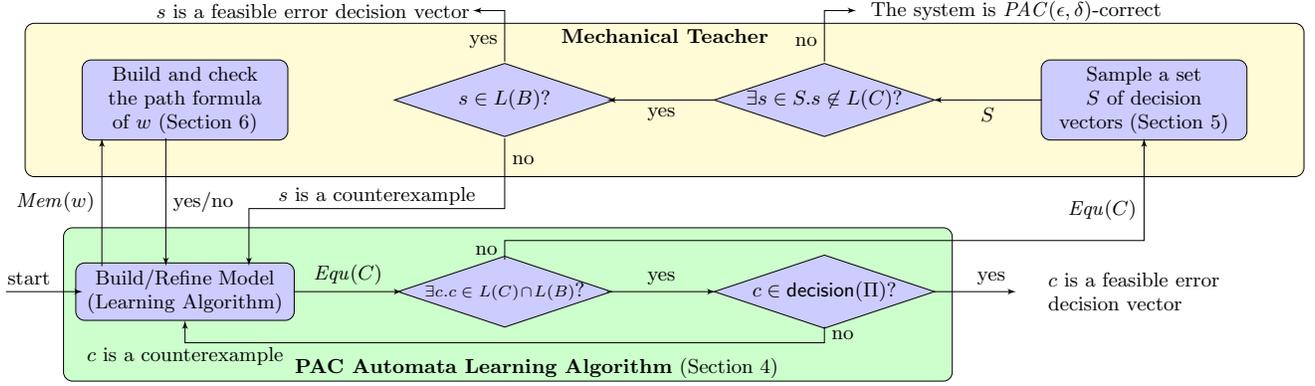
\begin{figure*}[t]
\scalebox{0.85}{

\begin{tikzpicture}[node distance = 2cm, auto]
\tikzstyle{decision} = [diamond, draw, fill=blue!20, aspect=3,text width=2.5cm, text badly centered, inner sep=0pt]
\tikzstyle{block} = [rectangle, draw, fill=blue!20, 
     text centered, rounded corners, minimum height=0.8cm]

\tikzstyle{line} = [draw, -latex']

  \draw[rounded corners,fill=green!20] (-1.9cm,1cm) rectangle +(13.9cm,-2.4cm);
	\node at (5.5,-1.2) {\textbf{PAC Automata Learning Algorithm} (Section~\ref{section:learning})};

    \node [block,text width=3cm] at (0,3) {Build and check the path formula of $w$ (Section~\ref{section:membership})};

	\node [rectangle, rounded corners, draw, fill=yellow!20,minimum width=20cm,minimum height=2.4cm] at (7.5,3){};
	\node at (7.5,4) {\bf Mechanical Teacher};

    \node [block,text width=3cm] at (0,3) {Build and check the path formula of $w$ (Section~\ref{section:membership})};

    \node [block,text width=3cm] at (15,3) {Sample a set $S$ of decision vectors (Section~\ref{section:sampling})};

	\node [decision]   at (10,3) {$\exists s\in S.  s \not\in L(C)$?};

	\node [decision]   at (5,3) {$s \in  L(B)$?};

    \node [block,text width=3.2cm] at (0,0) {Build/Refine Model (Learning Algorithm)};

	\node at (0, -1){$c$ is a counterexample};
	\node at (3.0, 1.5){$s$ is a counterexample};
	\node at (13,4.4){The system is $\pacofed$-correct};
	\node at (2,4.4){$s$ is a feasible error decision vector};

    \node [decision]  at (5,0) {{\scriptsize $\exists c. c\in L(C)\cap L (B)$}?\ };
	\node [decision]   at (10,0) {$c \in \DEC (\Pi)$?};
    \node [text width=3cm] at (15,0) {$c$ is a feasible error decision vector};
    \path [line] (-1.3,0.4) -- node {$\MEM(w)$} (-1.3,2.4);
    \path [line] (-0.3,2.4) -- node {yes/no} (-0.3,0.4);

    \path [line] (1.7,0) -- node {$\EQ(C)$} (3.4,0);
    \path [line] (6.6,0) -- node {yes} (8.3,0);

    \path [line] (10,-0.55) -- node {no} (10,-0.8) -- (0,-0.8) -- (0, -0.45);
    \path [line] (11.7,0) -- node[pos=0.7] {yes} (13,0);
    \path [line] (5,0.5) -- node {no} (5,0.8)--(15,0.8)--node[pos=0.3] {$\EQ(C)$}(15,2.4);

    \path [line] (13.4,3) -- node {$S$} (11.7,3);

    \path [line] (10,3.6) -- node[pos=0.4] {no} (10,4.4) -- (10.5,4.4);

    \path [line] (-2.8,0) -- node[pos=0.3] {start}  (-1.6,0);

    \path [line] (5,3.6) -- node[pos=0.4] {yes} (5,4.4) -- (4.5,4.4);
    \path [line] (8.3,3) -- node {yes} (6.6,3);
    \path [line] (5,2.4) -- node[pos=0.3] {no} (5,1.3) -- (1, 1.3)--(1,0.4);

\end{tikzpicture}
}
\vspace{-2mm}
\caption{A~detailed flow chart of our verification procedure}
\label{figure:flowchart}
\vspace{-5mm}
\end{figure*}

We summarize our procedure in this section. 
Let~$G$ be the CFG of the verified program, $k$ be the size of a batched
sample, $\epsilon$ be the error parameter, and $\delta$ be the confidence
parameter.
The goal of our procedure is to 
either find a feasible error decision vector of $G$ or show that $G$ is
$\pacofed$-correct.
In the latter case, we also
accompany our answer with a~$\pacofed$-correct regular representation
of the set of feasible decision vectors of $G$.
Let $\Pi$~be the set of feasible paths of~$G$, $D_k$ be the distribution defined
by our sampling mechanism
(cf.~Section~\ref{section:sampling}), and $\langof B$ be the
set of error decision vectors of $G$
(cf.~Section~\ref{section:error-branch-decisions}).

A~detailed
flow chart of our procedure can be found in Figure~\ref{figure:flowchart}.
First, the bottom part of the figure describes our learning algorithm.
We extend the online automata learning algorithm with two additional
tests for verification, as~described in Section~\ref{section:modify_learning}.
In particular, when the automata learning algorithm outputs a~candidate~$C$,
before sending teacher the equivalence query~$\EQ(C)$,
we first test whether $\langof C$ contains a~\emph{feasible} error decision vector~$c$.
In~case it does, we report~$c$ as an error.
Otherwise, in the case $c$ is both in $\langof C$ and $\langof B$ but is not feasible,
we return~$c$ to the learning algorithm to further refine the conjecture.


The top part of the figure describes our design of a~mechanical teacher.
The task of the teacher is to answer queries from the learning algorithm.
Membership queries of the form~$\MEM(w)$ can be answered by constructing the
path corresponding to the decision vector $w$ and the associated path formula,
which is then solved using a~constraint solver 
(cf.~Section~\ref{section:membership}).
Equivalence queries, on the other hand, are discharged using a~concolic tester
by checking whether there is a~decision vector~$s$ in the set of batched
samples~$S$ such that it does not belong to the language of~$C$
(cf.~Section~\ref{section:sampling}).
If no such a~decision vector exists, we conclude that the program is
$\pacofed$-correct.
Otherwise, we test whether $s \in \langof B$; if this holds, we report
that we have found a~feasible error decision vector.
In the case $s \notin \langof B$, it holds that $s$ is a~feasible decision
vector in $\DEC(\Pi)$ but not in the language of the current conjecture
$\langof C$.
If this happens, we return $s$ to the automata learning algorithm to refine the
conjecture and continue with the next iteration of the learning loop.

In general, our procedure is not guaranteed to terminate.
When the procedure terminates and reports an error (either by the teacher or the learning algorithm), a feasible
error decision vector is found and the program is reported to be incorrect.
If the teacher approves an approximate finite automaton $C$, our procedure
reports that $C$ is an approximate
model of $\DEC(\Pi)$ w.r.t.\ the $\pacofed$-correctness guarantee, which, in turn, implies that the program is $\pacofed$-correct.
From Lemma~\ref{lemma:mlstar}, we have the following theorem

\begin{theorem}
\label{theorem:procedure}
Let $\epsilon$ and $\delta$ be the error and confidence parameters respectively
and $\Pi$ be the set of feasible paths of a~given CFG~$G$.
If our procedure terminates with an approximate finite automaton~$C$, the
program is $\pacofed$-correct.
\end{theorem}
Moreover, we obtain the following corollary.
\begin{corollary}


Suppose our procedure reports a program $P$ is $\pacofed$-correct.
If we run the concolic tester with the same search strategy and batch size used
in our procedure on~$P$, with confidence~$\delta$, the concolic tester will
find an error with a~probability less than~$\epsilon$.
\end{corollary}

Thanks to the properties of the modified automata learning algorithm, when $\DEC(\Pi)$ is a regular set, our algorithm is guaranteed to terminate and either (1) return a~counterexample $c \in L(B)\cap \DEC(\Pi)$ or (2) find an approximate model of $\DEC(\Pi)$ that is disjoint with $L(B)$.

\section{Handling Procedure Calls}
\label{section:function}


\hide{In this section, we extend our formalism of CFGs to handle programs with
multiple procedures and procedure calls.
Moreover, we also refine our representation of error traces.
The issue of using finite automata to represent error decision vectors in the
said setting  is that when returning from a~procedure call, a~finite automaton
cannot remember an (unbounded) number of return points (in the case of
recursive procedures).
Therefore, an overapproximation, such as a~nondeterministic jump to any possible
return point, needs to be used.
The said overapproximation is, however, too imprecise and yields numerous
spurious errors.
To address this issue, we also describe our extension of the representation of
error decision vectors using pushdown automaton.
This representation of the set of error decision vectors in programs with
procedure calls is precise.}

In this section, we extend our formalism of CFGs to handle programs with
multiple procedures.
We use a PDA to represent error decision vectors in this setting.
The issue of using finite automata to represent error decision vectors in the
said setting is that when returning from a~procedure call, a~finite automaton
cannot remember an unbounded number of return points (in the case of
recursive procedures).
Therefore, an overapproximation, such as a~nondeterministic jump to any possible
return point, needs to be used.
The said overapproximation is, however, too imprecise and yields numerous
spurious errors.
In contrast, PDAs can represent the set of error decision vectors precisely.

On the other hand, we still use a~finite automaton to represent the
approximation of the set of feasible decision vectors $\DEC(\Pi)$.
As a consequence, except that we need to use PDA operations instead of FA operations and 
handle procedure calls in the membership queries, all other components remain unchanged for the setting of multiple procedures. 

\vspace{-0.0mm}
\subsection{Extending CFGs with Procedure Calls}\label{sec:label}
\vspace{-0.0mm}

\hide{
Consider a~set of \emph{procedure names} $\procs$
and a~set of procedure \emph{parameter names} $\params$, disjoint from
$\mathcal{X}$.
We use $\proccalls{\procs}[\mathcal{X}, \mathcal{F}]$ to denote the set of all
pairs $(p, g)$ where $p \in \procs$ is a~procedure name and $g$ is a~well-formed
first-order logic formula over $\mathcal{X}, \params$, and $\mathcal{F}$.
}

Assume the set of procedure names $\procs$.
A~\emph{CFG with calls} (CFGC) is defined as a~graph $G = (V, E, v_i, v_r, V_e,
\fp)$ where $V, v_i, v_r, V_e,$ and $\fp$ are defined in the same way as for
a~CFG, and
$E \subseteq (V \times \mathcal{T}[\mathcal{X}, \mathcal{F}] \times V) \cup
(V \times (\procs \times \mathcal{T}[\mathcal{X}, \mathcal{F}] \times \mathcal{T}[\mathcal{X}, \mathcal{F}]) \times V)$
is an extended set of edges that apart from local CFG edges $(v, f, v')$ for $f \in
\mathcal{T}[\mathcal{X}, \mathcal{F}]$ also contains \emph{procedure
call edges} $e = (v, (p, g_{\mathit{in}}, g_{\mathit{out}}), v')$ for $(p,
g_{\mathit{in}}, g_{\mathit{out}}) \in \procs \times \mathcal{T}[\mathcal{X},
\mathcal{F}] \times \mathcal{T}[\mathcal{X}, \mathcal{F}]$ and sequential nodes $v$.
The $g_{\mathit{in}}$ and $g_{\mathit{out}}$ components of $e$ correspond to
formulae for passing actual values to formal parameters of
$p$ (formula~$g_{\mathit{in}}$) and passing the return value of $p$ back to the
caller procedure (formula~$g_{\mathit{out}}$).

In this extension, we define a~\emph{program} as a~set of CFGCs $\prog = \{G_1,
\ldots, G_n\}$ together with a~(bijective) mapping $\cfg_{\prog}: \procs \to \prog$ that
assigns procedure names to~CFGCs.
We abuse notation and use $\prog$ to denote $\cfg_{\prog}$, i.e.,
$\cfgof{p}$ denotes the CFGC of a~procedure~$p$ in a~program $\prog$.
We assume that all CFGCs in $\prog$ have pairwise disjoint sets of nodes, and
further assume an entry point $\fun{main} \in \procs$.

In this paragraph, we give an informal description of how we extend the
definition of a~path from a program consisting of single CFG to a program consisting of a set of CFGCs and a dedicated entry point.
Given a procedure call edge $e$ in a~CFGC $G$, we call the \emph{inlining} of
$e$ in $G$ the CFGC $G'$ obtained from $G$ by substituting $e$ with the CFGC of
the called procedure.
We use $\semof{\prog}$ to denote the set of CFGs obtained from
$\cfgof{\fun{main}}$ by performing all possible (even recursively called)
sequences of inlinings, and removing any left procedure call edges from the
output CFGCs.
A~\emph{path} in $\prog$ is then a~sequence $\pi =
\pthof{v_0, f_1, v_1, f_2, v_2, \ldots, f_m, v_m}$ such that
there exists a~CFG~$G' \in \semof{\prog}$ for which it holds that $\pi \in
G'$.

\vspace{-0.0mm}
\subsection{Encoding Error Decision Vectors with Pushdown Automata}\label{sec:error-pda}
\vspace{-0.0mm}

In this section, we describe how we construct the PDA encoding the set of error
traces in the considered extension.
The general idea is the same as the one for the use of finite automata
(described in Section~\ref{section:error-branch-decisions}).
The main difference is that we add jumps between CFGCs (corresponding to
procedure call edges), which use the stack to remember which state the PDA
should return to after the procedure call terminates.

In the following, given a~CFGC $G = (V, E, v_i, v_r, V_e, \fp)$, we use $V(G),
E(G), \ldots, \fp(G)$ to denote the corresponding components of $G$, and,
moreover, we use $V_s(G)$ and $V_b(G)$ to denote the set of sequential and
branching nodes of $G$ respectively.
Consider a~program $\prog = \{G_1, \ldots, G_n\}$.
We construct the \emph{error path
automaton} as the PDA $B_P = (\{0, 1\}, Q, Q, q_i, \Delta, F)$ in the
following way:
\begin{itemize}
\vspace{-1mm}
  \itemsep0mm
  \item  $Q = V(G_1) \cup {} \cdots {} \cup V(G_n)$,
  \item  $q_i = v_i(G_k)$ such that $\cfgof{\fun{main}} = G_k$,
  \item  $F = V_e(G_1) \cup {} \cdots {}\cup V_e(G_n)$,
  \item  $\Delta = \Delta_1 \cup {} \cdots {} \cup \Delta_n$ where every $\Delta_j$ is
    defined as follows:
    \begin{itemize}
    \vspace{-1mm}
    \itemsep0mm
     \item $(v, 0, v'_0) \in \Delta_j$ if $v \in V_b(G_j)\setminus V_e(G_j)$, $(v, f_0, v'_0) \in E(G_j)$,
       and $v'_0$ is the $0$-successor of~$v$;
     \item $(v, 1, v'_1) \in \Delta_j$ if $v \in V_b(G_j)\setminus V_e(G_j)$, $(v, f_1, v'_1) \in E(G_j)$,
       and $v'_1$ is the $1$-successor of~$v$;
     \item $(v, \lambda, v') \in \Delta_j$ if $v \in V_s(G_j)\setminus
       V_e(G_j)$ and $(v, f, v') \in E(G_j)$;
     \item $(v, 0, v), (v, 1, v) \in \Delta_j$ if $v \in V_e(G_j)$; and
     \item $(v, [\lambda;\lambda/v'], v_i(G_k)), (v_r(G_k), [\lambda;v'/\lambda], v') \in \Delta_j$ if $v \in
       V_s(G_j)\setminus V_e(G_j)$, $(v, (p, g_{\mathit{in}}, g_{\mathit{out}}),
       v') \in E(G_j)$, and $G_k = \cfgof{p}$.
     \end{itemize}
\end{itemize}

\begin{lemma}
  Let $\prog$ be a~program, $\badpaths$ the set of error paths of $\prog$, and
  $B_P$ be the error path PDA for $\prog$. Then it holds that $\langof{B_P} =
  \DEC (\badpaths)$.
\end{lemma}

\section{Implementation}
\label{section:implementation}
We created a prototype tool \pacman that implements the verification procedure
described in this paper.
The tool uses several third-party libraries and tools.
First, it uses CIL (C Intermediate Language)~\cite{necula:cil2002} to
convert the verified C~program to a~set of~CFGCs, from
which we construct the error trace pushdown automaton $B_P$.
Further, we use the \libamore library~\cite{libamore,bollig:libalf2010} to perform operations of
automata, such as testing their membership and emptiness, or computing their
intersection.

For learning automata, we use the implementation of various learning algorithms
within the \libalf library~\cite{bollig:libalf2010}.
Membership queries are discharged using a~concolic tester, mentioned as an
alternative option in Section~\ref{section:membership}.
Given a decision vector, our tool uses the CFG of the program to generate
a~path corresponding to the decision vector.
The path is passed in the form of a~sequence of program statements to
the software model checker \cpachecker~\cite{beyer:cpachecker2011}, which
checks its feasibility.
It is possible to switch the model checker with other checkers, such as
\cbmc~\cite{clarke:tool04}.

To deal with equivalence queries, we modified the concolic tester
\crest~\cite{burnim:heuristics2008} to generate a~batch of $k$ decision vectors, as described in Section~\ref{section:sampling}.
As \crest may fail to generate the decision vector of a program
execution when the execution terminates abnormally, we modified \crest
to take a finite prefix of the execution in this case.
One issue of \crest that we encountered is that when it processes a condition
composed using Boolean connectives, it expands the condition into a~cascade of
\texttt{if} statements corresponding to the Boolean expression, making the
program longer and harder to learn.
We addressed this by modifying \crest so that it can process conditions with
Boolean connectives inside without expanding them, and in this way we increased
the performance and precision of the analysis.

We also implemented the following three optimizations to improve the
performance of the prototype.

\paragraph{Intersection with Bad Automaton}

Recall that our modified learning algorithm (described in
Section~\ref{section:modify_learning}) first checks whether the intersection of
the language of the conjecture $L(C)$ and the bad language $L(B_P)$ is empty.
Checking emptiness of a~PDA is, however, more difficult than that of
a~finite automaton.
To speed up the procedure, we build a finite automaton $B_O$ that
over-approximates the error language and always first checks whether $L(C)\cap
L(B_O)=\emptyset$, which is an emptiness test for finite automata.
We check $L(C)\cap L(B_P)=\emptyset$ only for the cases that the
previous test fails.

\paragraph{Counterexample from the Learning Algorithm}

When an equivalence query returns a counterexample $c$, automata
learning algorithms usually do not guarantee that $c$ is not a valid
counterexample in the next conjecture automaton.
In our preliminary experiments, we found out that it happens very often that the
mechanical teacher returns the same counterexample in several consecutive
iterations.
Therefore, we decided to check whether $c$ is still a~valid counterexample (by
a~membership query) for the learning algorithm before proceeding to the
emptiness test.
In the case $c$ is valid, it will be immediately returned to the learning
algorithm to refine the conjecture.

\paragraph{Handling Membership Queries}

The main bottleneck of our approach is the time spent for membership
queries.
In our implementation, the software model checker
\cpachecker is used to check whether a path is feasible.
For each membership query, if we invoke \cpachecker with a system
call, a~Java virtual machine will be created and the components of \cpachecker
need to be loaded, which is very time consuming.
To make membership queries more efficient, we modified \cpachecker
to run in a~server mode so that it can check more than a~single path without
being re-invoked.

\section{Experiments}
\label{section:experiments}

This section presents our experimental results to justify the
claims made in this paper.
We evaluated the performance of our prototype using the recursive
category of SV-COMP~2015~\cite{svcomp} as the benchmark.
The recursive category consists of 24 non-trivial examples such as
Ackermann, McCarthy 91, and Euclidean algorithms. Among eight
participating tools, only two can solve 20 or more examples correctly. 
\hide{
 and is one
of the most difficult categories in the competition~\footnote{The winner
of the category received 27 out of 40 points.}.
}
Among the 24 examples, 8 of them contain an error.
We performed our experiments with the error parameter $\epsilon = 0.1$, confidence
$\delta = 0.9$, and size of batched samples $k = 10$.
We ran our prototype on each example three times in all experiments.
The provided statistical data were calculated based on the average of the
three runs unless explicitly stated otherwise.
We set the timeout to 900\,s to match the rules of SV-COMP~2015.

\vspace{-0.0mm}
\subsection{Comparison of Learning Algorithms}\label{sec:label}
\vspace{-0.0mm}

We evaluated our approach with different automata learning algorithms
implemented within the \libalf library.
There are five active online automata learning algorithms implemented in
\libalf: Angluin's $L^*$~\cite{angluin:learning1987}, $L^*$-columns,
Kearns/Vazirani (KV)~\cite{kearns:introduction1994}, Rivest/Schapire
(RS)~\cite{rivest:inference1993}, and
$\mathit{NL}^*$~\cite{bollig:angluin2009}.
Among the search strategies provided by \crest, we chose the
\emph{random branch} strategy.
The experimental results are in Table~\ref{table:cmp_learning_alg}.

\begin{table}[tb]
\centering
\caption{Comparison of learning algorithms}
\vspace{-2mm}
\scalebox{0.9}{
\begin{tabular}{|l|l|l|l|l|l|}
\cline{2-6}
\multicolumn{1}{c|}{} & \multicolumn{5}{c|}{Algorithms}\\
\cline{2-6}
\cline{2-6}
\multicolumn{1}{c|}{}&KV&$L^*$&$L^*$-col.&RS&$NL^*$\\
\hline
Verified                 &   15 & 9.67 &    10 & 11.33 &  8.33 \\
Bug found                &    6 & 6.33 &     6 &  6.33 &     6 \\
\ \ by Bad               &    4 & 4.67 &     4 &     3 &  4.33 \\
\ \ by \crest            &    2 & 1.67 &     2 &  3.33 &  1.67 \\
\hline
False positives           &    0 &    0 &     1 &  0.33 &     1 \\
False negatives           &    2 & 1.67 &  1.33 &  1.67 &     1 \\
Timeouts                  &    1 & 6.33 &  5.67 &  4.33 &  7.67 \\
\hline
\# of $\MEM$ queries      & 2896 & 8898 & 10071 & 15377 & 14463 \\
\# of $\EQ$ queries       &  548 &   78 &    77 &   367 &    67 \\
\hline
Total time [s]           & 2406 & 6668 &  6565 &  5786 &  7972 \\
$\MEM$ queries time      & 30\,\% & 59\,\% &  58\,\% &  63\,\% &  70\,\% \\
\hline
\end{tabular}
}
\label{table:cmp_learning_alg}
\vspace{-2mm}
\end{table}

The results show that KV is the algorithm with the
best performance---it solved 21 out of the 24 examples.
Our technique solves more than any participant in the recursive
category of SV-COMP~2015 but the winner.
The main reason for the performance difference is that KV uses a
tree-based data structure to store query results.
Compared to other learning algorithms that use table-based structures,
KV requires much less number of membership queries to maintain the
consistency of the tree-based structure.
For all learning algorithms except RS, the number of error paths found
by the emptiness test of the intersection of the conjecture and the bad
automaton is more than that found by \crest.
In our experiments, the time spent for membership queries is usually
the performance bottleneck.
Table~\ref{table:cmp_learning_alg} shows that membership queries took
30\,\% of the total execution time for KV and at least 58\,\% for
other algorithms.

\vspace{-0.0mm}
\subsection{Comparison of Search Strategies}\label{sec:label}
\vspace{-0.0mm}

We also evaluated the performance of our algorithm against different \crest
search strategies.
According to~\cite{burnim:heuristics2008}, the most
efficient ones are \emph{random branch} strategy (RBS for short) and
\emph{control-flow directed} strategy (CDS for short).
Therefore, we tested the performance of our prototype using these two
strategies.
We selected KV as the learning algorithm in this experiment.
The results are shown in Table~\ref{table:cmp_crest_strategies}.

\begin{table}[tb]
\centering
\caption{Comparison of search strategies of \crest}
\vspace{-2mm}
\begin{tabular}{|l|l|l|}
\cline{2-3}
\multicolumn{1}{c|}{} & \multicolumn{2}{c|}{Search Strategy}\\
\cline{2-3}
\cline{2-3}
\multicolumn{1}{c|}{}&\emph{RBS}&\emph{CDS}\\
\hline
Verified                &   15 &   15 \\
Bug found               &    6 &    6 \\
\hline
False positives         &    0 &    0 \\
False negatives         &    2 &    2 \\
Timeouts                &    1 &    1 \\
\hline
\# of $\MEM$ queries    & 2896 & 2362 \\
\# of $\EQ$ queries     &  548 &  463 \\
\hline
Total time [s]          & 2406 & 2013 \\
Time for one sample [s] & 0.75 & 0.82 \\
\hline
\end{tabular}
\label{table:cmp_crest_strategies}
\vspace{-4mm}
\end{table}

Table~\ref{table:cmp_crest_strategies} shows that although the average
time for taking one sample with CDS is more than with RBS, the total
time is less.
The main reason is that CDS explores untouched branching points more
aggressively than RBS but requires more overhead.
Our experiments conform the results
in~\cite{burnim:heuristics2008}.

\vspace{-0.0mm}
\subsection{Evaluation of CREST with Restarts}\label{sec:label}
\vspace{-0.0mm}
\begin{table}[htb]
\centering
\caption{Evaluation of \crest with and without restart.  For each example, the
number of batches and the number of iterations used to find bugs are obtained
respectively from the worst run in scenario~(1) and from the best run in
scenario~(2).}
\scalebox{0.95}{
\begin{tabular}{|l|l|l|}
\cline{2-3}
\multicolumn{1}{c|}{} & \multicolumn{2}{c|}{Settings} \\
\hline
Examples & \emph{Batch Size 10} & \emph{Never Restart} \\
\hline
Ackermann02           & batch 3 & iteration 14 \\
Addition02            & batch 1 & iteration 2 \\
Addition03            & Timeout & Timeout     \\
BallRajamani-SPIN2000 & batch 1 & iteration 1 \\
EvenOdd03             & batch 1 & iteration 2 \\
Fibonacci04           & batch 4 & iteration 2 \\
Fibonacci05           & Timeout & Timeout     \\
McCarthy91            & batch 1 & iteration 2 \\
\hline
\end{tabular}
}
\label{table:cmp_crest_restart}
\end{table}
To justify our modification to the $\pacofed$-correctness guarantee given in Section~\ref{sec:gen-stoch-eq},
we show in the experiment below that running \crest in batches does
not decrease its bug-hunting capabilities.
We compared the performance of \crest with two different scenarios:
(1)~restart after each 10 decision vectors and
(2) never restart.
We performed the experiment on the 8 buggy examples in the recursive
category and calculated in how many examples \crest found a~bug for
within the timeout period.
In~Table~\ref{table:cmp_crest_restart}, we chose RBS as the search strategy.
We also tried the experiments with the CFG strategy and got a similar result.
We list the worst result for
scenario~(1) and the best result for scenario~(2) that we received in our three
runs.
We found out that the worst runs in scenario~(1) can still find
with a~little overhead all bugs found by the best runs in scenario~(2).

\vspace{-0.0mm}
\subsection{\hspace*{-1mm}Evaluating Quality of Learned Automata}\label{sec:label}
\vspace{-0.0mm}
Besides the performance in terms of the running time, we also compared the
quality of the learned automata produced by our prototype using the
two strategies for the 15 successfully verified bug-free examples.
To evaluate the quality of the learned automata, for each example, we
ran \crest with the given search strategy to get 100 batched samples and 
tested how many of the them are all accepted by the learned
automaton.
The average values of the runs are shown in
Table~\ref{table:conj_quality1} where evaluation strategies are
strategies used to generate the testing batched samples.
The table shows that the quality of the automata
learned with the two strategies is almost the same.
Also, observe that the guarantee of our procedure is that the sample
coverage is higher than 90\,\%.
Our experimental results show that the quality of the automata
produced by our procedure matches the theoretical expectations.

\begin{table}[tb]
\centering
\caption{Comparison of \crest search strategies in terms of the quality of
  the learned automata. The total number of tested batches is 1500.}
\begin{tabular}{|l|l|l|l|}
\cline{3-4}
\multicolumn{2}{c|}{} & \multicolumn{2}{c|}{$\EQ$ query strategy}\\
\cline{3-4}
\cline{3-4}
\multicolumn{2}{c|}{}&\emph{RBS}&\emph{CDS}\\
\hline
Evaluation     & Accepted  & 1487    & 1473 \\
\cline{2-4}
(RBS)          & Ratio     & 99.13\,\% & 98.2\,\% \\
\hline
\hline
Evaluation     & Accepted  & 1489    & 1500 \\
\cline{2-4}
(CDS)          & Ratio     & 99.27\,\% & 100\,\% \\
\hline
\end{tabular}
\label{table:conj_quality1}
\vspace{-2mm}
\end{table}

Finally, we tested how many words generated by \crest
are not covered in the automata learned with the KV algorithm and RBS.
Again, we ran \crest with RBS in two scenarios:
(1)~restart after each 10 decision vectors and
(2)~never restart.
For each learned automaton (there were 15) and each scenario, we generated 1000
decision vectors and checked how many of them are accepted by the
automaton.
In total, for scenario~(1), we observed 1487 accepted batches of size 10 (for
the total of 15\,000 tested vectors), yielding the correctness 99.13\,\%.
For scenario~(2), we observed 14\,977 accepted vectors, for the correctness
99.86\,\%.
We notice that no matter which strategy we use, the learned automaton accepts
over 99\,\% of the decision vectors produced by \crest.


\section{Discussion}
\label{section:conclusions}
There are several advantages of having a program model with statistical guarantees.
For instance, the model can be reused for verifying a~different set of properties of the program.
Assume that the new property to be verified is described as an error path
automaton $B'$ and $C$~is the learned automaton.
If $\langof {B'} \cap \langof C =\emptyset$, we verified the program with the
new property and the same $\pacofed$-correctness guarantee.
For the case that there exists a~decision vector $w\in \langof {B'} \cap
\langof C$, we test whether $w$~is feasible and either report that $w$ is
a~feasible error decision vector w.r.t.\ $B'$ or continue the learning algorithm
with~$w$ as a~counterexample for refining the next conjecture.

In this paper we focus on checking validity of program assertions.
The verification step is handled by making an intersection of the conjecture
automaton $C$ and the error path automaton $B$ and testing its emptiness.
This procedure can be generalized to more sophisticated safety properties by
replacing the tests $\langof C \cap \langof B =\emptyset$ and $s \in \langof B$
with other tests.
For example, we can check the property ``the program contains at most 10
consecutive 1-decisions on any path'' with a~statistical guarantee of the
correctness of the received answer.
By extending the alphabet $\{0,1\}$ with program labels, one can also check temporal
properties related to those labels, e.g., ``label $A$ should be reached within
10 decisions after label $B$ is reached.''

One possible extension of our work is to learn sequences of feasible function
calls instead of decision vectors.
This might lead to a more compact model in contrast to the current approach.
However, in this case the alphabet of the model to be learned will be all
function names in the program, which is usually much larger than 2, the size of
the alphabet in our work.
Moreover, in this setting, it is much harder to answer membership queries;
a~program path composed of function calls
might perform a~complex traversal through loops and branches in between the
calls, making the problem of checking feasibility of a~program path already
undecidable.

One benefit of our approach is that, in principle, it can be extended to black
box system verification and model synthesis.
By observing the behavior of the environment, we may find some pattern (e.g., some
statistical distribution) of the inputs and then, based on that, design
a~sampling mechanism.
Under the assumption that the behavior of the environment remains unchanged, we can verify or
synthesize the model of the system w.r.t.\ the given sample distribution. 

\hide{The design of the sampling mechanism is one of the most difficult part in the entire procedure. Our first proposal was to give all decision vectors with length smaller than $k$ the sample probability to be chosen, where $k$ is a big integer number. However, we found that most of the decision vectors are infeasible. In our experience, in it very often that we can claim an automaton with empty language is 99\% similar to $\DEC(\Pi)$ using this sampling mechanism. 

\td{OL: remove the previous paragraph? it doesn't suit me there ...}}


\section{Related Works}
\label{section:related}


%

Exact automata learning algorithm was first proposed by
Angluin~\cite{angluin:learning1987} and later improved by many
people~\cite{angluin:learning1987,rivest:inference1993,kearns:introduction1994,bollig:angluin2009}.
The concept of probably approximately correct (PAC) learning was first proposed
by Valiant in his seminal work~\cite{valiant:theory1984}.
The idea of turning an exact learning algorithm to a PAC learning algorithm can
be found in Section~1.2 of~\cite{angluin:queries1987}.

Applying PAC learning to testing has been considered before~\cite{walkinshaw:assessing2011,goldreich:property1998}. 
The work in \cite{goldreich:property1998} considers a program that manipulates
graphs and check if the output graph of the program has properties such as
being bipartite, $k$-colorable, etc.
Our work considers assertion checking, which is more general than the
specialized properties.
The work~\cite{walkinshaw:assessing2011} considers more
theoretical aspects of the problem.
The author estimates the maximal number of queries required to infer a~model of
a~black box machine.
The context is quite different, e.g., the work does not discuss how to sample
according to some distribution efficiently to produce the desired guarantee
(bounded path coverage) as we do in this paper.

The $L^*$ algorithm has been used to infer the model of error traces of a program.
In~\cite{chapman:learning2015}, instead of decision vectors, the authors try to
learn the sequences of function calls leading to an error.
Their teacher is implemented using a~bounded model checker and hence can only guarantee correctness up to a~given bound.
The authors do not make use of the PAC learning technique as we did in this work.

Both our approach and statistical model
checking~\cite{sen:statistical2004,legay:statistical2010,zuliani:bayesian2010}
provide statistical guarantees.
As mentioned in the introduction, statistical model checking assumes a given
model while our technique generates models of programs with statistical guarantee. Those
models can be analyzed using various techniques and reused for
verifying different properties.


\newcounter{newsection}
\setcounter{newsection}{\value{section}}
\addtocounter{newsection}{1}

\appendix

\vspace{-0.0mm}
\section{An Inlining in a~CFGC}\label{app:inlining}
\vspace{-0.0mm}

Consider a~procedure call edge $e = (v, (p, g_{\mathit{in}},
g_\mathit{out}), v')$ of
$G_j = (V_j, E_j, v_{ij}, v_{rj}, V_{ej}, \fp{}_j)$ for $1 \leq j \leq n$.
Further assume that 
$\cfgof{p} = G_k = (V_k, E_k, v_{ik}, v_{rk}, V_{ek}, \fp{}_k)$ for $1 \leq k
\leq n$.
The \emph{inlining} of $e$ in $\prog$ is the program $\prog^{\#} = \{G_1, \ldots,
G_{j-1}, G^{\#}_j, G_{j+1}, \ldots, G_n\}$ such that
%
$G^{\#}_j = (V_j \cup V^{\#}_k, E^{\#}, v_{ij},
v_{rj}, V_{ej} \cup V^{\#}_{ek}, \fp{}_j)$
%
where $V^{\#}_k$ is a~set of fresh nodes (i.e., $V^{\#}_k$ is disjoint from the
set of
nodes of any CFGC in $\prog$) and $V^{\#}_{ek} \subseteq V^{\#}_k$.
Moreover, there exist bijections $\subst: V^{\#}_k \to V_k$ and $\subst_e:
V^{\#}_{ek}
\to V_{ek}$ such that $\subst_e \subseteq \subst$.
The set of edges $E^{\#}$ is defined as
%
  $E^{\#} = (E_j \setminus \{e\}) \cup E^{\#}_k \cup \{(v, g_{\mathit{in}},
  v^{\#}_{ik}), (v^{\#}_{rk}, g_{\mathit{out}}, v')\}$
%
where $E^{\#}_k =  \{(v^{\#}_k, f, v^{\#\prime}_k) \mid (\substof{v^{\#}_k}, f,
\substof{v^{\#\prime}_k}) \in E_k\}$.

Let $\inlngsof{\prog}$ be the smallest (potentially infinite) set that contains
$\prog$ and is closed w.r.t.\
inlinings in $\fun{main}$, i.e., if $\inlngsof{\prog}$ contains a~program
$\prog^{\#}$ with a~procedure call edge $e$ in the CFGC
$\cfgprimeof{\fun{main}}$, it also contains the inlining of~$e$ in~$\prog^{\#}$.
We abuse notation and use $\subst$ uniformly to denote for all programs
$\prog^{\#} \in \inlngsof{\prog}$ the mapping of the nodes of the CFGCs
in~$\prog^{\#}$ to their original nodes in $\prog$ (for nodes $V$ of $\prog$, we
assume that $\subst|_{V} = \mathrm{id}$, i.e., that the restriction of $\subst$
to $V$ is the identity relation).
We~denote as $\semof{\prog}$ the set of CFGs (without procedure call edges)
obtained by starting from the set $\inlngsof{\prog}$,
collecting the CFGCs of $\fun{main}$ functions of all inlinings of $\prog$
into the set $M = \{\cfgprimeof{\fun{main}} \mid \prog^{\#} \in
\inlngsof{\prog}\}$, and, finally, transforming the CFGCs of $M$ into
CFGs of $\semof{\prog}$ by substituting each procedure call edge $(v^{\#}, (p^{\#},
g^{\#}_{\mathit{in}}, g^{\#}_{\mathit{out}}), v^{\#\prime})$ with the edge
$(v^{\#}, \FF, v^{\#\prime})$.

We extend the definition of a~path as follows:
A~\emph{path} in $\prog$ is a~sequence $\pi =
\pthof{v_0, f_1, v_1, f_2, v_2, \ldots, f_m, v_m}$ such that
there exists a~CFG~$G' \in \semof{\prog}$ and a~CFG path
$\pi' = \pthof{v'_0, f_1, v'_1, f_2, v'_2, \ldots, f_m, v'_m}$ in~$G'$ such that
$\forall 0 \leq j \leq m: v_j = \substof{v'_j}$.

\def\thesection{\arabic{newsection}}

\bibliographystyle{plain}
\bibliography{refs}

\end{document}